\documentclass{llncs}

\usepackage{tikz}
\usepackage{amsmath,amssymb}

\newcommand {\N}	  {\mathbb{N}}

\usepackage{enumerate}

\newcommand{\comment}[1]{}
\newcommand{\ignore}[1]{}

\def\clap#1{\hbox to 0pt{\hss#1\hss}}

\makeatletter
  \def\moverlay{\mathpalette\mov@rlay}
  \def\mov@rlay#1#2{\leavevmode\vtop{%
    \baselineskip\z@skip \lineskiplimit-\maxdimen
    \ialign{\hfil$#1##$\hfil\cr#2\crcr}}}
\makeatother

\newcommand{\remove}[1]{}

\begin{document}
\title{Node-balancing  by edge-increments} 

\author{Friedrich Eisenbrand\inst{1}  and 
Shay Moran\inst{2,3} and Rom Pinchasi\inst{2} and Martin Skutella\inst{3}  }

\institute{EPFL, 1015 Lausanne, Switzerland 
\and
Technion, Israel Institute of Technology, Haifa 32000, Israel. 
\and
Max Planck Institute for Informatics, Saarbr\"{u}cken, Germany.
\and
Technische Universit{\"a}t Berlin, Stra\ss{}e des 17. Juni 136, 
10623 Berlin, 
Germany
}

\maketitle

\begin{abstract}
\noindent 
Suppose you are given a graph $G=(V,E)$ with a
weight assignment $w:V\rightarrow\mathbb{Z}$ and that
your objective is to modify $w$ using legal steps
such that all vertices will have the same weight,
where in each legal step you are allowed to choose an
edge and increment the weights of its end points by $1$.

In this paper we study several variants of this problem
for graphs and hypergraphs. 
On the combinatorial side we show connections
with fundamental results from matching theory
such as Hall's Theorem and Tutte's Theorem.
On the algorithmic side 
we study the computational complexity of associated decision
problems.

Our main results are a characterization of the graphs for which any
initial assignment can be balanced by edge-increments and a strongly
polynomial-time algorithm that computes a balancing sequence of
increments if one exists. 

\end{abstract}

\begin{figure}
  \centering  

  \tikzstyle{vertex}=[circle,fill=black!25,minimum size=15pt,inner sep=0pt]
  \tikzstyle{selected vertex} = [circle,fill=red!25,minimum size=15pt,inner sep=0pt]
  \tikzstyle{weight} = [font=\small]
  \tikzstyle{selected edge} = [draw,thick, ->,gray!50,decorate, decoration=snake]
  \tikzstyle{ignored edge} = [draw,line width=5pt,-,black!20]
  
  \begin{tikzpicture}


\def \n {6}
\def \radius {2cm}
\def \margin {8} 

\foreach \s in {1,...,\n}
{
  \node[vertex] (\s) at ({360/\n * (\s - 1)}:\radius) {$\s$};
  \draw[ >=latex] ({360/\n * (\s - 1)+\margin}:\radius) 
    arc ({360/\n * (\s - 1)+\margin}:{360/\n * (\s)-\margin}:\radius);
}



\end{tikzpicture}
\hspace{1cm}
 \begin{tikzpicture}


\def \n {6}
\def \radius {2cm}
\def \margin {8} 

\foreach \s in {3,...,\n}
{
  \node[vertex] (\s) at ({360/\n * (\s - 1)}:\radius) {$\s$};
  \draw[ >=latex] ({360/\n * (\s - 1)+\margin}:\radius) 
    arc ({360/\n * (\s - 1)+\margin}:{360/\n * (\s)-\margin}:\radius);
}

\node[selected vertex] (1) at ({360/\n * (1 - 1)}:\radius) {$2$};
  \draw[ >=latex] ({360/\n * (1 - 1)+\margin}:\radius) 
    arc ({360/\n * (1 - 1)+\margin}:{360/\n * (1)-\margin}:\radius);

\node[selected vertex] (2) at ({360/\n * (2 - 1)}:\radius) {$3$};
  \draw[ >=latex] ({360/\n * (2 - 1)+\margin}:\radius) 
  arc ({360/\n * (2 - 1)+\margin}:{360/\n * (2)-\margin}:\radius);



\end{tikzpicture}

  \caption{The node-weights after one incrementing step.}
  \label{fig:1}
\end{figure}
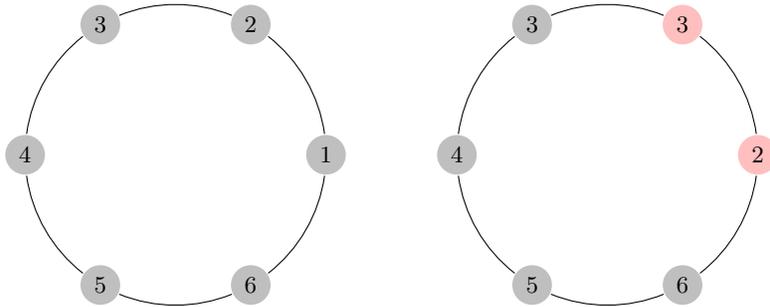

\section{Introduction}
The following puzzle is often used as an introductory puzzle for the method
of invariance and potential functions: Six boxes numbered $1$ to $6$ 
are arranged in a cycle. For every $1 \leq i \leq 6$, we start with 
$i$ oranges in box number $i$. At each step we are allowed to add one orange
to each of two adjacent boxes. Prove that we will never be able to make all
boxes contain the same number of oranges.

One of the simple solutions to this puzzle is to observe that the total number
of oranges in boxes $1,3,5$ is always smaller than the total number of oranges 
in boxes $2,4,6$ and this never changes through each step of the game.

In this paper we consider the natural generalization of the puzzle above to 
arbitrary graphs. 
Let $G=(V,E)$ be a finite graph, let $w:V\rightarrow\mathbb{N}$ be a non-negative 
integer weight function on its vertices and let $e=\{u,v\}\in E$.  A
\emph{positive step} on $e$ modifies $w$ by increasing the weights of
$u,v$ by $1$ unit.  We say that $w$ is
\emph{equatable~in~$G$} 
if there exists a sequence of only positive steps, $S=s_1,\dots,s_m$,
after which all vertices have the same weight.  We also say that the
sequence $S$ \emph{positively equates} $w$.

\medskip

\noindent
Our \emph{main results} are the following. 

\begin{enumerate}[i)]
\item We characterize those graphs $G = (V,E)$ for which any initial
  assigmnet $w : V \rightarrow \N_0$ is equatable. These are the
  connected graphs with an odd number of nodes for which $G - U$ has less than
  $|U|$ isolated vertices for any $U \subset V$. Here $G-U$ is the
  subgraph of $G$ that is induced by $V \setminus U$.
  (Theorem~\ref{theorem:graph-pos}) \label{item:1}
\item We show that the following problem can be solved in strongly polynomial time. Given a graph $G = (V,E)$ and an initial assignment $w: V \rightarrow \N_0$, decide whether $w$ is equatable and compute an equating multiset of edges. (Theorem~\ref{thm:pos-equi-poly}) \label{item:2}
\item An initial assignment $w$ of the nodes of a  bipartite graph $G = (L+R,E)$  is not equatable if $w(L) \neq w(R)$, the difference $w(L)-w(R)$ is invariant under edge-increments. However, each balanced assignment with $w(L) = w(R)$ is equatable if and only if the strict Hall condition holds: For any nonempty set of vertices $X$ that is properly contained in $L$ or in $R$, one has $|X|< |N(X)|$. Here $N(X)$ denotes the neighborhood of $X$. (Theorem~\ref{theorem:bipartite-pos}) \label{item:3}
\item Finally we show that the analog of the decision problem \ref{item:2}) is $NP$-hard for hypergraphs. (Theorem~\ref{theorem:hyper-alg-pos}). 
\end{enumerate}



\subsubsection*{Related work} 

The problem of equating the node-weights is closely related to \emph{perfect $b$-matchings},~\cite{Schrijver03}. Let $b \in \N_o^{|V|}$ be a vector of non-negative node-weights.  A  $b$-matching of a graph $G = (V,E)$ is a vector $x \in \N_0^{|E|}$ that satisfies 
\begin{equation}
  \label{eq:1}
  \sum_{e \in \delta(v)} x_e \leq  b_v,
\end{equation}
where  $\delta(v)$ denotes the set of edges of $G$ that are incident to $v$. 
A $b$-matching is \emph{perfect}, if the inequality in~\eqref{eq:1} can be replaced by equality. 
Thus $b$-matchings are a generalization of \emph{matchings}, where $b$ is the all ones vector. 

What is the relationship between $b$-matchings and the process of equating positive weights in graphs by edge-increments? 
Suppose that the given initial weight assignment $w \in \N_0^{|V|}$ is equatable and that the resulting equated node-weight is $\beta \in \N$.  Then, the edge-increments that lead to the balanced node-weight $\beta$ are a $b$-matching $x \in \N_0^{|E|}$ with $b_v = \beta - w_v$ for each vertex $v$. By incrementing the node-weights of each edge $e$ exactly $x_e$ times, one arrives at a balanced assignment with weight $\beta$ on all the nodes.

Maximum weight $b$-matchings can be computed in polynomial
time~\cite{Edmonds65,Edmonds65b,cunningham1978primal}.  The currently
fastest algorithms for maximum weight matching are by
Gabov~\cite{gabow1990data}, and Gabov and
Tarjan~\cite{gabow1991faster}. The fastest algorithm for weighted
$b$-matching is by Anstee~\cite{anstee1987polynomial}. Recent exciting
progress for maximum cardinality matching has been given by
Madry~\cite{madry2013navigating} improving upon the $O(m \sqrt{n})$
running time of Hopcroft and Karp~\cite{hopcroft1973n} and
\cite{karzanov1973finding} in the sparse case.

A related notion to equatable graphs is the one of a \emph{regularizable} graph. A graph is regularizable, if there exists a $k$ and a perfect $k$-matching such that each edge is chosen at least once in this matching. Thus, one obtains a $k$-regular graph by replacing each edge by as many parallel edges, as its multiplicity in the $b$-matching. Berge~\cite{berge1978regularisable} provided the following characterization of regularizable graphs. If $G$ is connected an bipartite, then $G$ is regularizable if and only if $|N(U)|>|U|$ for each non-empty stable set $U$ of $G$. This is a \emph{strict} Hall condition for stable sets.

\section{A characterization of equatable graphs}
\label{sec:char-equit-graphs}

Which are the graphs $G=(V,E)$ for which any initial assignment $w:V
\rightarrow \N_0$ is equatable? The following theorem provides the
answer to that question.

\begin{theorem}\label{theorem:graph-pos}
Let $G=(V,E)$ be a finite graph.
The following statements are equivalent:
\begin{enumerate}
\item Every integer assignment $w: V \rightarrow \N_0$ is  equatable  in $G$. \label{item:4}
\item $G$ is connected, $|V|$ is odd and for all  $ U\subseteq V$,  the graph $G-U$ has less 
than $|U|$ isolated vertices,
where $G-U$ is the vertex induced subgraph on $V-U$. \label{item:5}
\end{enumerate}
\end{theorem}

Notice that condition~\ref{item:5}) implies that $G$ has at least $3$ vertices. 
We will now provide the proof of this theorem. To do so, we rely on a well known result of Tutte that characterizes the existence of a perfect $b$-matching in a graph. 

\begin{theorem}[Tutte~\cite{tutte1952factors}]\label{theorem:aux}
Let $G=(V,E)$ be a finite graph and let $b: V \rightarrow \N_0$ be
a weight function on the vertices of $G$. The following statements are equivalent. 

\begin{enumerate}[a)]
\item $G$ has a perfect $b$-matching. 
\item For every (possibly empty) subset $U$ of $V$ 
\begin{equation}\label{eq:aux}
\sum_{x \in U}b(x) \geq \sum_{x \in I(U)}b(x) + S(G-U),
\end{equation}
where $I(U)$ is the set of isolated vertices of $G-U$ and  
$S(G-U)$ is the number of connected components of $G-U$ that are not
isolated vertices whose total $b$-weight is odd.
\end{enumerate}
\end{theorem}

\begin{proof}[Proof of Theorem~\ref{theorem:graph-pos}]
Suppose that condition~\ref{item:5}) holds. In order to show that any $w \in \N_0^{|V|}$ is equatable, it is enough to show that each assignment $w^{(v)}$ with 
\begin{displaymath}
  w^{(v)}(u) =
  \begin{cases}
    1, & \text{ if } u=v \\
    0, & \text{ otherwise}
  \end{cases}
\end{displaymath}
is equatable, since the corresponding steps decrease the weight of $v$ relative to the other vertices by exactly one. We show this by establishing existence of   a perfect $b$-matching with $b(v) = 2\cdot n$ and $b(u) = 2 \cdot n +1$ for each other vertex $u \neq v$.  

Let $U \subseteq V$. We have to show~\eqref{eq:aux}. 
If $U=\emptyset$, then $G-U=G $. 
The total $b$-weight of $G$ is even and $G-\emptyset$ has only one component, since $G$ is connected, thus $S(G-U) = 0$. Also, $G - \emptyset$ does not have isolated vertices. This shows that the  right-hand-side of  \eqref{eq:aux}
is $0$.

If $U \neq \emptyset$ then, by our assumption, $|I(U)|\leq |U|-1$.
We have 
\begin{equation}\label{eq:a1}
\sum_{x \in U}b(x) \geq (2n+1)|U|-1.
\end{equation}
Indeed, there is equality in (\ref{eq:a1}) only if $v \in U$.

On the other hand $\sum_{x \in I(U)}b(x) \leq |I(U)|(2n+1)$.
Therefore, 
\begin{equation}\label{eq:a2}
\sum_{x \in I(U)}b(x) \leq |I(U)|(2n+1) \leq (|U|-1)(2n+1).
\end{equation}

Finally, the term $S(G-U)$ is at most the number of components of $G$
that are not isolated vertices. Therefore,

\begin{equation}\label{eq:a3}
S(G-U) \leq \frac{n-1}{2}.
\end{equation}

Inequality (\ref{eq:aux}) is, therefore, satisfied
because using  (\ref{eq:a1}), (\ref{eq:a2}), and (\ref{eq:a3}) inequality
(\ref{eq:aux}) reduces to 
$$
(2n+1)|U|-1 \geq (|U|-1)(2n+1) +\frac{n-1}{2},
$$  
which clearly holds.

Suppose now that every $w \in \N_o^{|V|}$ is equatable. Then clearly,  $G$ is connected. Also $G$ has an odd number of vertices since the parity of the sum of weights is invariant under the edge-increment operation. In a graph with an even number of nodes, an equated assignment has even parity which shows that an odd initial assignment is not equatable. 

Let $U$ be any non-empty set of vertices of $G$. Assume to the contrary that
$G-U$ has $k \geq |U|$ isolated vertices. Denote by $I$ the set of isolated
vertices in $G-U$ and let $v$ be a fixed vertex in $U$. 
Consider the weights assignment $w:V \rightarrow \N_0$ such that 
$w(v)=1$ and for any other vertex $v' \in V$ we have $w(v')=0$.
We reach a contradiction by showing that $w$ is not  equatable. 
To see this observe that
any  step that increases by $1$ the weight of a vertex in $I$ must 
increase by $1$ the weight of some vertex in $U$.
It follows that at any moment the sum of the weights of the vertices in $I$
is strictly smaller than the sum of the weights of the vertices in $U$.
Because $|I| \geq |U|$, it is not possible to reach a situation where 
all vertices in $I \cup U$ have the same weight. \qed
\end{proof}





\section{A polynomial-time algorithm to equate the weights}
\label{sec:polyn-time-algor}

In this section, we deal with the computational problem of deciding
whether an initial assignment $w: V \rightarrow \N_0$ is equatable
and, if so, how to compute a multiset of edges that leads to such
equated weights. Let us recall the connection to the $b$-matching problem. If we know a number $\beta \in \N$  such that all node-weights can be brought to $\beta$ by increment-steps, then the multiset of edges leading to uniform weights $\beta$ is a perfect $b$-matching with weights $b(v) = \beta - w(v)$ for each $v \in V$. The primary question is then: Can  $\beta$ be efficiently computed? We will now give a positive answer to this question. The main result of this section is the following theorem. We first provide an upper bound on $\beta$. 

\begin{theorem}\label{thm:pos-equi-poly}
Given a graph $G=(V,E)$ and an integer weights assignment $w:V\rightarrow\N_0$, one can determine in strongly polynomial time whether $w$ is positively equatable in $G$. Moreover, the smallest multiset of edges equating~$w$ can be determined efficiently.
\end{theorem}

We again make use of Theorem~\ref{theorem:aux}. For some  target value~$\beta\geq\max_{v\in V}w(v)$ let~$b_{\beta}:V\to\mathbb{Z}_{\geq0}$ with~$b_{\beta}(v):=\beta-w(v)$. By Theorem~\ref{theorem:aux} there is a sequence of positive steps starting from node weights $w$ and yielding uniform node weight~$\beta$ if and only if~\eqref{eq:aux} holds for $b=b_{\beta}$ and all subsets~$U$ of~$V$. 

\begin{lemma}\label{lem:pos-equi-poly}
If~$w$ is equatable, there is such a value~$\beta$ that is bounded from above by~$n\max_{v\in V}w(v)$.
\end{lemma}

\begin{proof}
For the trivial case where $w$ is uniformly zero we can choose $\beta=0$. Thus, in what follows we might assume that $\max_{v\in V}w(v)\geq1$. Notice that with respect to~\eqref{eq:aux} the only subsets $U$ of $V$ that might force~$\beta$ to be big are those with~$|U|>|I(U)|$ (otherwise, if $|U|\leq|I(U)|$, the left hand side of~\eqref{eq:aux} as a function of $\beta$ increases at most as fast as the right hand side does). For such subset $U$, however, and for $\beta=n\max_{v\in V}w(v)$ we get
\begin{align*}
\sum_{x\in U}b_{\beta}(x)&\geq |U|(\beta-\max_{v\in V}w(v))\\
&\geq|U|\beta-n\max_{v\in V}w(v)+n-|U| \\
& =(|U|-1)\beta+n-|U|\geq\sum_{x\in I(U)}b_{\beta}(x)+S(G-U)
\end{align*}
which concludes the proof. \qed
\end{proof}

\begin{proof}[Proof of Theorem~3]  
From now on we fix the parity of $\beta$ (even or odd) such that $S(G-U)$ no longer depends on the particular value of $\beta$ (in our algorithm we deal with the two cases sequentially). In particular, for a fixed subset $U$ of $V$, both the left hand side and the right hand side of~\eqref{eq:aux} are linear functions of~$\beta$. Therefore, for each $U\subseteq V$, one of three cases holds:
\begin{itemize}
\item[(i)] \eqref{eq:aux} is satisfied for all values of $\beta$ or for no value of $\beta$;
\item[(ii)] there is a $\beta_U\in\mathbb{Z}$ such that \eqref{eq:aux} is satisfied if and only if $\beta\geq\beta_U$;
\item[(iii)] there is a $\beta_U\in\mathbb{Z}$ such that \eqref{eq:aux} is satisfied if and only if $\beta\leq\beta_U$.
\end{itemize}
This observation finally enables us to find the smallest feasible
value of $\beta$ (with fixed parity) by binary search in polynomial
time: Let $\alpha=\max_{v\in V}w(v)$ and $\gamma=n\max_{v\in
  V}w(v)$. Due to Lemma~\ref{lem:pos-equi-poly}, we can restrict our
search for a suitable value~$\beta$ to the
interval~$[\alpha,\gamma]$. For fixed~$\beta'\in[\alpha,\gamma]$, we
can test in polynomial time whether~\eqref{eq:aux} is satisfied for
all subsets $U$ of $V$ and obtain a violating subset $U$ in the
negative case. In fact such such a violating set is found by the algorithm for the perfect $b_{\beta'}$-matching problem 
that also certifies the non-existence of a perfect
$b_{\beta'}$-matching with a set $U \subseteq V$
violating~(\ref{eq:aux}).

In the positive case, we can decrease the upper bound~$\gamma$ to~$\beta'$ and continue the search. In the negative case, we distinguish the three cases (i), (ii), and (iii) listed above w.r.t.\ the violating subset~$U$. In case~(i), there is no feasible~$\beta$ and we thus terminate the search. In case~(ii) we obtain a new lower bound~$\beta_U>\beta'$ and thus continue the search after replacing $\alpha$ with $\beta_U$. Finally, in case~(iii) we obtain a new upper bound~$\beta_U<\beta'$ and thus continue the search after replacing~$\gamma$ with $\beta_U$.

Notice that the running time of the resulting binary search algorithm is only weakly polynomial. A strongly polynomial running time can be achieved by replacing binary search with \emph{parametric search}~\cite{Megiddo79}.\qed
\end{proof}

\begin{remark}  
  We sketch Megiddo's parametric search
  technique~\cite{Megiddo79} in our setting. Suppose that we want to
  find the smallest even $\beta$ such that there exists a perfect
  $b_\beta$-matching.  Consider a fully combinatorial algorithm~$A$
  for the non-parametric problem, that is, for the perfect
  $b$-matching problem. A fully combinatorial algorithm uses only
  additions, subtractions, and comparisons.  In fact, such an
  algorithm exists for the perfect $b$-matching problem, if the
  parities of the $b(v)$ are fixed, as it is for the case of all
  parametric $b_\beta$, if $\beta$ is even,
  see~\cite[p.~186]{Gerards94}. More precisely, the algorithm consists then
  of solving one general network flow problem and one perfect matching
  problem on graphs that are polynomial in the size of the graph $G$
  that is in our input, see, \cite{Gerards94}. For both subproblems,
  there exist fully-combinatorial algorithms, for example the minimum
  mean-cycle algorithm~\cite{Goldberg:1989:FMC:76359.76368} and
  Edmond's algorithm~\cite{Edmonds65}.

  Algorithm~$A$ is now modified in order to solve the parametric
  problem. For this, the modified algorithm has to work with linear
  functions of the parameter $\beta$ instead of just constant
  numbers. Notice that adding or subtracting two linear functions
  yields a linear function again. Comparing two linear functions,
  however, imposes a problem. Whenever algorithm~$A$ compares two
  numbers, the modified version first has to determine whether the
  desired value $\beta_{OPT}$ is smaller or larger than the unique point
  $\beta^*$ at which the two linear functions cross (if at all). This
  can be decided by calling any algorithm~$B$ for the perfect
  $b$-matching problem as a subroutine for the fixed parameter value
  $\beta^*$. 
  Depending on the outcome, the corresponding alternative of the comparison is chosen, for example in an {\tt if}-conditional, and one continues to run algorithm $A$ for the parametric problem.  
The number of calls of $B$ is bounded by the number of
  comparisons performed by~$A$ which is strongly polynomial in the
  input size. In this way, finding the desired value $\beta_{OPT}$ is
  reduced to a series of non-parametric $b$-matching problems.
\end{remark}

\section{Bipartite graphs.}
\label{sec:bipartite-graphs}

Going back to the elementary puzzle presented at the beginning, 
observe that the corresponding graph $G$  is bipartite  where the two parts
have the same cardinality. The initial weight $w$ is not equatable, since the sum of the weights of the vertices in one part of the bi-partition is not equal to the sum of the weights. 

Let $G=(L,R,E)$ be a bipartite graph. 
An assignment of weights $w$ to the vertices of $G$ is called 
\emph{balanced} if $w(L)=w(R)$, 
where for a subset $U$ of vertices, $w(U)$ 
is defined as $\sum_{v\in U}{w(v)}$.

We now  characterize those bipartite graphs
for which all balanced assignment $w$ are equatable.

\begin{theorem}\label{theorem:bipartite-pos}
Let $G=(L+R,E)$ be a bipartite graph. The following statements are equivalent:
\begin{enumerate}[i)]
\item Every balanced assignment is positively equatable in $G$. \label{item:6}
\item For any non empty set of vertices $X$ that is properly 
contained either in $L$ or in $R$, $|X| < |N(X)|$. \label{item:7}
\end{enumerate}
\end{theorem}
Here $N(X)$ denotes the \emph{neighborhood} of $X$, that is, 
the set of vertices in $G$ that are neighbors of
some vertex in $X$.  Notice that  condition~\ref{item:7}) implies that $|L| = |R|$ holds. Condition~\ref{item:7})  is
a ``strict'' version of the well known \emph{Hall's condition} for the existence
of a perfect matching. A bipartite graph has a perfect matching 
if and only if for any non empty set of vertices $X$ that is properly
contained either in $L$ or in $R$, $|X| \leq |N(X)|$.

\begin{proof}

Suppose that every balanced assignment to the vertices of $G$ is
positively equatable. 
Assume to the contrary that there exists $X \subset L$, $0<|X|<|L|$, 
and $|X| \geq |N(X)|$ (the symmetric case where $X \subset R$ is similar).
If $N(X)=\emptyset$, then the vertices in $X$ are isolated and 
any balanced assignment of weights to the vertices 
in $G$ where
the vertices in $X$ get weight $0$ and some other vertex not in $X$ gets 
a positive weight is not equatable.

If $N(X)$ is not empty, 
consider the following balanced assignment of weights to the vertices of $G$.
Pick a vertex $v \in L \setminus X$ and a vertex $u \in N(X)$. We define
$w(v)=w(u)=1$ and for every other vertex $z$ of $G$ we define $w(z)=0$.
Clearly, $w$ is balanced. However, the graph $G$ together
with the assignment of weights $w$ is not equatable.
This is because a positive step increases the total weight of the vertices in
$N(X)$ by at least the same amount by which it increases the total weight of 
the vertices in $X$. If after a series of positive steps the weights of the
vertices in $G$ are the same, then in particular the total weight
of the vertices in $X$ is at least as large as the total weight of the vertices
in $N(X)$ (because $|X| \geq |N(X)|$), but this is impossible because 
for the initial assignment of weight the total weight of the vertices in 
$X$ is $0$ while the total weight of the vertices in $N(X)$ is $1$.

Now, suppose that~\ref{item:7}) holds.  As in the proof of
Theorem~\ref{theorem:graph-pos}, it is enough to show that any
assignment with $w(u) = w(v) = 1$ for some $u \in L$ and $v \in R$ and
$w(x)=0$ for any other vertex is equatable, since then, each balanced
assignment can be equated. This however follows from the fact that $G
- \{u,v\}$ has a perfect matching, since it satisfies Hall's
condition. \qed
\end{proof}

\subsection{Hypergraphs}
\label{sec:hypergraphs}

One can naturally  generalize the equating problem to hypergraphs. 
In this setting, one is given a hypergraph  $H=(V,E)$ and an integer weights assignment
$w:V\rightarrow\N_0$. A \emph{a positive step}
on $e\in E$ modifies $w$ by increasing  the
weights of each $u\in E$ by $1$ unit. The rest of the definitions
are generalized in the obvious way. Not surprisingly, 
deciding, whether one can equate the weights in a hypergraph is 
 NP-complete. This follows by a reduction from \emph{$3$-dimensional matching}~\cite{Karp72}.  

 \begin{figure}
   \centering
   \begin{tikzpicture}
    \draw (3,4) ellipse (2cm and 1cm);
    \node at (3,4) {$H$};
    \draw[rotate around ={-45:(2,2)}] (2,2) ellipse (2cm and .5cm);
    \draw[rotate around ={45:(4,2)}] (4,2) ellipse (2cm and .5cm);
    \filldraw [lightgray](3,1.1) circle (2pt) node[right,black] {$1$};
    \filldraw [lightgray](1.2,2.8) circle (2pt) node[right,black] {$1$};
    \filldraw [lightgray](4.8,2.8) circle (2pt) node[right,black] {$1$};
\end{tikzpicture}
   \caption{NP-completeness in the hypergraph setting. }
   \label{fig:2}
 \end{figure}
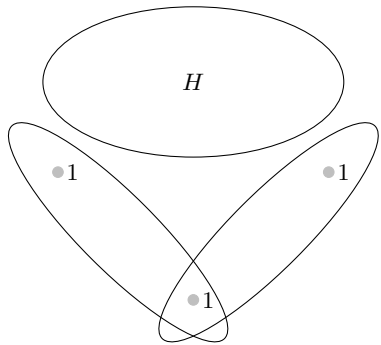

Thus deciding whether a hypergraph has a perfect matching is an  NP-complete problem. This can be trivially reduced to the equating problem by adding three new vertices and two new edges, each consisting of two of the three new vertices. The three vertices have weight $1$ while all other vertices have weight $0$. If these weights can be equated, then they all have weight $1$ in an equated assignment. Thus, the weights can be equated if and only if the original hypergraph has a perfect matching. Consequently, the following theorem holds.

\begin{theorem}\label{theorem:hyper-alg-pos}
The decision problem of determining for
any hypergraph $H=(V,E)$ and any integer weights assignment
$w:V\rightarrow\mathbb{Z}$, whether
$w$ is positively equatable in $H$ is
NP-complete. 
\end{theorem}


\end{document}